\documentclass[conference]{IEEEtran}

\usepackage{cite}
\usepackage{amsmath,amsthm,amssymb,amsfonts}
\usepackage{graphicx}
\usepackage{textcomp}
\usepackage{xcolor}
\usepackage{color}
\usepackage{setspace}

\usepackage{mathrsfs}
\usepackage{pifont}
\usepackage{bm}
\usepackage{pgfplots}
\usepackage{tikz}
\usetikzlibrary{arrows}
\usepackage{subfigure}
\usepackage{graphicx,booktabs,multirow}
\usepackage{upgreek}
\usepackage{bm}
\usepackage{mathrsfs}
\usepackage{algorithm}
\usepackage{algorithmicx}
\usepackage{algpseudocode}
\newcommand{\vct}[1]{{\bm #1}}

\definecolor{colorhkust}{RGB}{20,43,140}
\definecolor{colortsinghua}{RGB}{116,52,129}
\definecolor{color1}{RGB}{128,0,0}

\newtheorem{theorem}{Theorem}

\newtheorem{proposition}{Proposition}

\newtheorem{example}{Example}

\newcommand{\trace}{{\rm Tr}}

\mathchardef\re="023C
\mathchardef\im="023D

\IEEEoverridecommandlockouts
\begin{document}

\title{How Neural Architectures Affect Deep Learning for Communication Networks?}

%

\author{
   \IEEEauthorblockN{Yifei Shen$^\dagger$, \textit{Student Member}, \textit{IEEE}, Jun Zhang$^\dagger$, \textit{Fellow}, \textit{IEEE}, and Khaled B. Letaief$^{\dagger\ddagger}$}, \textit{Fellow}, \textit{IEEE}\\
  \IEEEauthorblockA{$^\dagger$Dept. of ECE, The Hong Kong University of Science and Technology, Hong Kong\\
  	$^\ddagger$ Peng Cheng Laboratory, Shenzhen, China\\
  	Email:{ yshenaw@connect.ust.hk, eejzhang@ust.hk, eekhaled@ust.hk} 
  	\thanks{This work was supported by the Hong Kong Research Grants Council under Grant No. 16210719 and 15207220.}}

}

\maketitle

\begin{abstract} In recent years, there has been a surge in applying deep learning to various challenging design problems in communication networks. The early attempts adopt neural architectures inherited from applications such as computer vision, which suffer from poor generalization, scalability, and lack of interpretability. To tackle these issues, domain knowledge has been integrated into the neural architecture design, which achieves near-optimal performance in large-scale networks and generalizes well under different system settings. This paper endeavors to theoretically validate the importance and effects of neural architectures when applying deep learning to communication network design. We prove that by exploiting permutation invariance, a common property in communication networks, graph neural networks (GNNs) converge faster and generalize better than fully connected multi-layer perceptrons (MLPs), especially when the number of nodes (e.g., users, base stations, or antennas) is large. Specifically, we prove that under common assumptions, for a communication network with $n$ nodes, GNNs converge $O(n \log n)$ times faster and their generalization error is $O(n)$ times lower, compared with MLPs. 
\end{abstract}

\begin{IEEEkeywords}
Communication networks, deep learning, graph neural networks, neural tangent kernel.
\end{IEEEkeywords}

\section{Introduction}\label{sec:intro}
Deep learning has recently emerged as a revolutionary technique for various applications in communication networks, e.g., resource management \cite{sun2018learning}, data detection \cite{he2020model}, and joint source-channel coding \cite{shao2021learning}. The early attempts adopt neural architectures inherited from applications such as computer vision, e.g., fully connected multi-layer perceptrons (MLPs) or convolutional neural networks (CNNs) \cite{sun2018learning,shen2018lora,eisen2019learning}. Although these classic architectures achieve near-optimal performance and provide computation speedup for small-scale networks, the performance is severely degraded when the number of clients becomes large. For example, for FDD massive MIMO beamforming, the performance of CNNs is near-optimal for a two-user network while an $18\%$ gap to the classic algorithm exists with $10$ users \cite{ma2021neural}. Moreover, these neural architectures generalize poorly when the system settings (e.g., the signal-to-noise ratio or the user number) in the test dataset are different from those in the training dataset. For example, for power control in cloud radio-access networks (Cloud-RANs), the performance degradation of MLPs can be more than $50\%$ when the SNR in the test dataset is slightly different from that in the training dataset \cite{shen2018lora}. Dense communication networks, a characteristic of 5G, usually involve hundreds of clients, and the user number and SNR change dynamically. Hence, it is very challenging to apply MLP-based methods in practical wireless networks.

To improve scalability and generalization, recent works incorporated the domain knowledge of target tasks to improve the architectures of neural networks, e.g., unrolled neural networks \cite{he2019model,shi2021algorithm,monga2021algorithm}, and group invariant neural networks \cite{keriven2019universal,shen2020graph}. Particularly, graph neural networks (GNNs) have recently attracted much attention thanks to their superior performance on large-scale networks \cite{lee2019graph,eisen2019optimal,shen2020graph,jiang2020learning,satorras2021neural,chowdhury2021unfolding,guo2021learning}. GNNs achieve good scalability, generalization, and interpretability by exploiting the permutation invariance property in communication networks \cite{shen2020graph,eisen2019optimal,guo2021learning}. For example, for the beamforming problem, a GNN trained on a network with $50$ users is able to achieve near-optimal performance on a network with $1000$ users \cite{shen2020graph}. In \cite{jiang2020learning}, GNNs were applied to resource allocation without channel state information (CSI) in intelligent reflecting surface (IRS) aided systems. It was shown that GNNs not only generalize well across different SNRs and different numbers of users, but the reflecting angles generated by GNNs are also easy to interpret.

Despite the empirical successes, it remains elusive why these architectures outperform unstructured MLPs and how much performance gains we can obtain via improving the neural architecture. Meanwhile, for reliable operation in real systems, it is crucial to provide theoretical guarantees and understand when the neural network works. For deep learning-based methods, the test performance, e.g., the sum rate or spectrum efficiency, is controlled by the \emph{convergence rate} at the training stage and the \emph{generalization error} at the test stage. Specifically, if the convergence speed is too slow, we cannot even obtain a low training loss, let alone performing well at the test stage. Additionally, combining the generalization error and training error provides an upper bound for the test performance. Unfortunately, the existing theoretical analysis of MLPs and GNNs for communication networks \cite{sun2018learning,qiang2019deep,shen2020graph} cannot characterize either of them.

The technical difficulty of the theoretical analysis lies in the non-convex nature of neural networks. Fortunately, there is a recent breakthrough that connects overparameterized neural networks and neural tangent kernels (NTK), which makes the training objective convex in the functional space \cite{jacot2018neural}. Based on this result, we will demonstrate the importance and effects of neural architectures when applying deep learning
to designing communication networks. Specifically, we observe that the convergence and generalization are determined by the alignment between the eigenvectors of the \emph{random} NTK matrix and the label vectors (Theorems \ref{thm:ntk_conv} and \ref{thm:gen_all}). Thus, we theoretically characterize how neural architectures affect convergence and generalization by studying this alignment (Theorems \ref{thm:conv} and \ref{thm:gen}). Specifically, we prove that under certain assumptions, for a communication network with $n$ nodes, GNNs converge $O(n \log n)$ times faster and the generalization error is $O(n)$ times lower, compared with MLPs. This demonstrates that GNNs are superior to MLPs in large-scale communication networks. To the best of our knowledge, this paper is the first attempt to theoretically study the convergence benefits of structured neural networks from both the communication and machine learning communities. Hence the developed results are not only timely but also will have significant impacts on the design, analysis, and performance evaluation of communication networks.

\section{Preliminaries}\label{sec:pre}
\subsection{Permutation Invariance}
Let $[n]$ represent the set $\{1, \cdots, n \}$, and denote the permutation operator as $\pi:[n] \rightarrow [n]$. For the vector-form variable $\bm{\gamma} \in \mathbb{C}^d$ and matrix-form variable $\bm{\Gamma} = [\bm{\gamma}_1,\cdots,\bm{\gamma}_{n}] \in \mathbb{C}^{n \times d}$, the permutation  is defined as 
\begin{align*}
(\pi \star \bm{\gamma})_{ (\pi(i_1)) } = \bm{\gamma}_{ (i_1) }, \quad (\pi \star \bm{\Gamma})_{ (\pi(i_1),:) } = \bm{\Gamma}_{ (i_1,:) }. 
\end{align*} 
A function $f(\cdot)$ is called \emph{permutation invariant} if for any $\pi$, we have $f(\pi \star \bm{\gamma} ) = f(\bm{\gamma} )$ or $f(\pi \star \bm{\Gamma} ) = f(\bm{\Gamma})$ when the input is a vector or matrix, respectively.

\begin{example}\label{exp:sum} (Sum and weighted sum) It is easy to check that the sum function is permutation invariant as $\sum_{i=1}^n x_i = \sum_{i=1}^n x_{\pi(i)}$. The weighted sum $\sum_{i=1}^n w_i x_i$ is not permutation invariant if the vector variable is $\bm{\gamma} = [x_i]_{i=1}^n$. However, if the variables are $\bm{\gamma}_i = [w_i, x_i]$, then we have 
\begin{align*}
    f(\bm{\Gamma}) = \sum_{i=1}^n w_i x_i = \sum_{i=1}^n w_{\pi(i)} x_{\pi(i)} = f(\pi(\bm{\Gamma})),
\end{align*}
which makes the weighted sum permutation invariant.
\end{example}

\subsection{Permutation Invariant Problems in Communication Networks}
In this paper, we consider the following permutation invariant optimization problem,
\begin{equation}\label{eq:opt_pro}
\begin{aligned}
&\mathscr{P}:\underset{\bm{\Gamma}}{\text{minimize}}
& & g(\bm{\Gamma})
& \text{ subject to }
& & Q(\bm{\Gamma}) \leq 0
\end{aligned},
\end{equation}
such that $g(\bm{\Gamma}) = g(\pi \star \bm{\Gamma})$, $Q(\bm{\Gamma}) = Q(\pi \star \bm{\Gamma}), \forall \pi$. 

We next present the power control problem in a $K$-user interference channel as a specific example. Let $p_{k}$ denote the transmit power of the $k$-th transmitter, $h_{k,k} \in \mathbb{C}$ denote the direct-link channel between the $k$-th transmitter and receiver, $h_{k,j} \in \mathbb{C}$ denote the cross-link channel between transmitter $j$ and receiver $k$, $s_k \in \mathbb{C}$ denote the data symbol for the $k$-th receiver, and $n_k \sim \mathcal{CN}(0,\sigma_k^2)$ is the additive Gaussian noise. The signal-to-interference-plus-noise ratio (SINR) for the $k$-th receiver is given by $\text{SINR}_k = \frac{|h_{k,k}|^2p_k}{\sum_{i\neq k}|h_{k,i}|^2p_i+\sigma_k^2}$. The power control problem is formulated as follows:
\begin{align*}
&\underset{\vct{p}}{\text{maximize}}
& & \sum_{k=1}^{K} w_k \log_2 \left(1+ \text{SINR}_k \right) \\
& \text{subject to}
& & 0 \leq p_k \leq 1, \forall k,
\end{align*}
To elaborate the permutation invariance property of this problem, we consider a permuted problem with parameters $w'_{\cdot}, h'_{\cdot, \cdot}, p'_{\cdot}, \sigma_{\cdot}$, such that $p'_k = p_{\pi(k)}, w'_k = w_{\pi(k)}$, $h'_{k,i} = h_{\pi(k),\pi(i)}, \sigma_k = \sigma'_{\pi(k)}$.

Under this permutation, we have $\sum_{i\neq k} |h_{k,i}|^2 p_k + \sigma_k^2 = \sum_{i \neq \pi(k)} |h'_{\pi(k),i}|^2 p_{\pi(k)} + (\sigma'_{\pi(k)})^2$ and thus $\text{SINR}_{k} = \text{SINR}'_{\pi(k)}$. We then have
\begin{align*}
    \sum_{k=1}^{K} w_k \log_2 \left(1+ \text{SINR}_k \right) &= \sum_{k=1}^{K} w'_{\pi(k)} \log_2 \left(1+ \text{SINR}'_{\pi(k)} \right) \\
    &\overset{(a)}{=} \sum_{k=1}^K w'_{k} \log_2 \left(1+ \text{SINR}'_{k} \right),
\end{align*}
where \emph{(a)} is due to the permutation invariant propery of the weighted sum as shown in Example \ref{exp:sum}. The following proposition shows that permutation invariant problems are ubiquitous in communication networks.

\begin{proposition}\label{fact:graph_pi} \cite{shen2020graph}
Any graph optimization problem can be formulated as in  \eqref{eq:opt_pro}.
\end{proposition} 

A direct implication of Proposition \ref{fact:graph_pi} is that if the problem can be formulated as a graph optimization problem, then it enjoys the permutation invariance property. As communication networks can naturally be modeled as graphs, resource allocation in communication networks can be formulated as graph optimization problems. Examples include the $K$-user interference channel beamforming (modeled in Section II.C of \cite{shen2020graph}), joint beamforming and phase shifter design in IRS-assisted systems (modeled in Section IV.A of \cite{jiang2020learning}), and power control in multi-cell-multi-user systems (modeled in Section II.A of \cite{eisen2019optimal} and Section III of \cite{guo2021learning}). Additionally, inference on factor graphs is also a graph optimization problem. Thus, channel estimation or data detection also enjoys the permutation invariance property \cite{satorras2021neural}. Furthermore, graph structures are ubiquitous in signal processing systems, e.g., topological interference management, hybrid precoding, localization, and traffic prediction. As a result, permutation invariance also holds for these problems.

\subsection{Message Passing Graph Neural Networks}
To apply deep learning to solve Problem \eqref{eq:opt_pro}, our task is to identify a neural network that maps the problem parameters to the optimal solution. Thus, it is desirable that the adopted neural architecture respects the permutation invariance property of the problem. Message passing graph neural networks (MPGNNs), which are developed for learning on graphs, are a class of neural networks that exploit the permutation invariance property. Like other neural networks, they adopt a layer-wise structure. The update rule for the $k$-th layer at vertex $i$ in an MPGNN is
\begin{equation}\label{eq:mpgnn}
\begin{aligned}
 \bm{x}_i^{(k)} = \alpha^{(k)}\left(\bm{x}_i^{(k-1)}, \phi^{(k)} \left(\left\{\left[\bm{x}_j^{(k-1)},\bm{e}_{j,i}\right]: j \in \mathcal{N}(i)   \right\} \right)  \right),
\end{aligned}
\end{equation}
where $\bm{x}_i^{(0)}$ is the node feature of node $i$, $\bm{e}_{j,i}$ is the edge feature of the edge $(j,i)$, $\mathcal{N}(i)$ is the set of neighbors of node $i$, and $\bm{x}_i^{(k)}$ is the hidden state of node $i$ at the $k$-th layer. If the desired output is a single vector, then the output of the MPGNN is given by $\bm{o} =  \sum_{i=1}^n \bm{x}^{(K)}_i$, where $K$ is its maximal layer, and $n$ is the number of nodes in the graph. If the desired output is a vector for each node, the output of MPGNNs is given by $\bm{O} =  \left[\bm{x}^{(K)}_1, \cdots, \bm{x}^{(K)}_{n}\right]^T$.

\begin{figure*}[htb]
	\centering
	\subfigure[$5$ users.]{
		\begin{minipage}[t]{0.35\linewidth}
			\centering
			\includegraphics[width=1\linewidth]{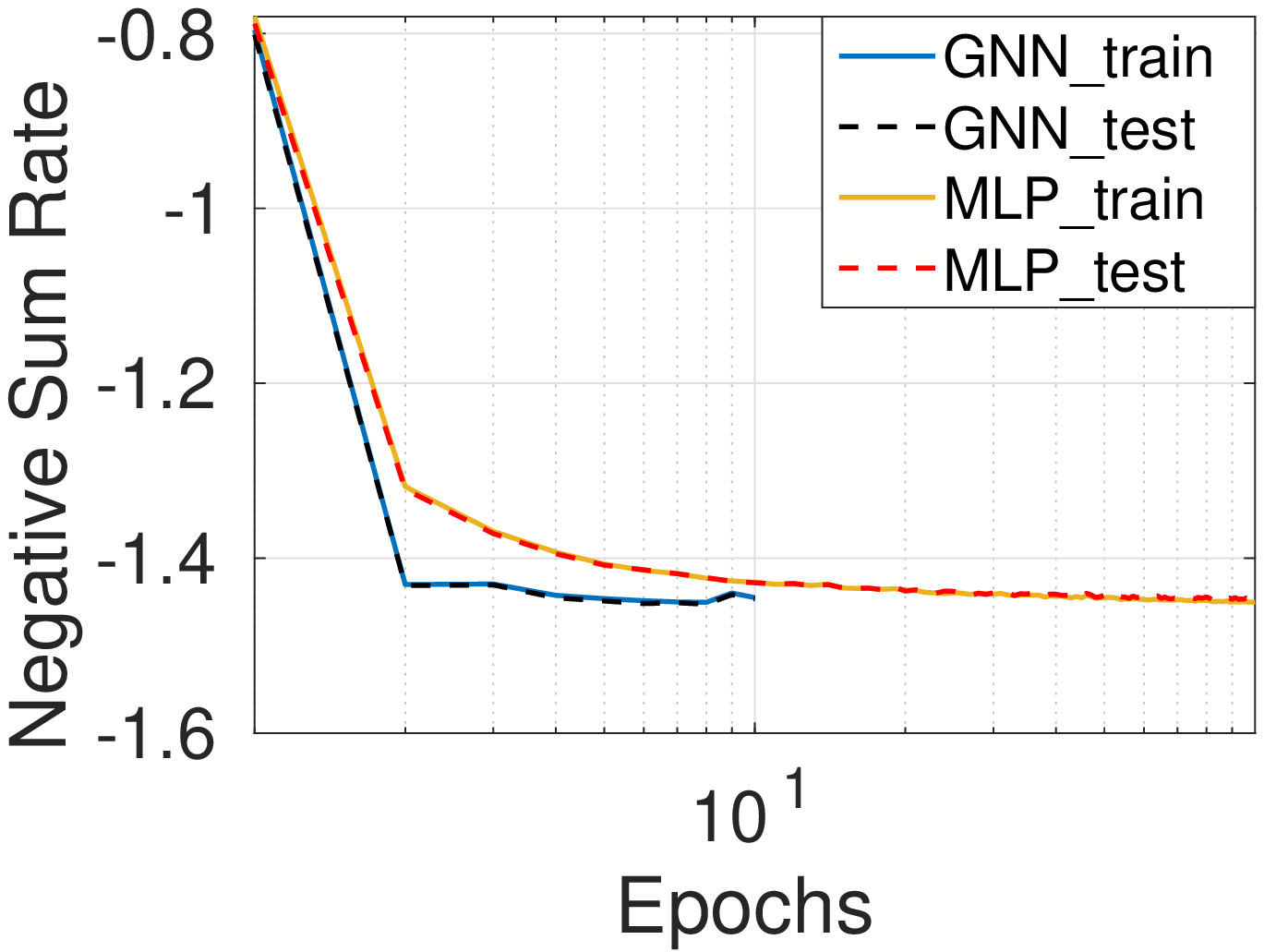}
			\label{fig:e2e}
		\end{minipage}%
	}%
	\subfigure[$20$ users.]{
		\begin{minipage}[t]{0.35\linewidth}
			\centering
			\includegraphics[width=1\linewidth]{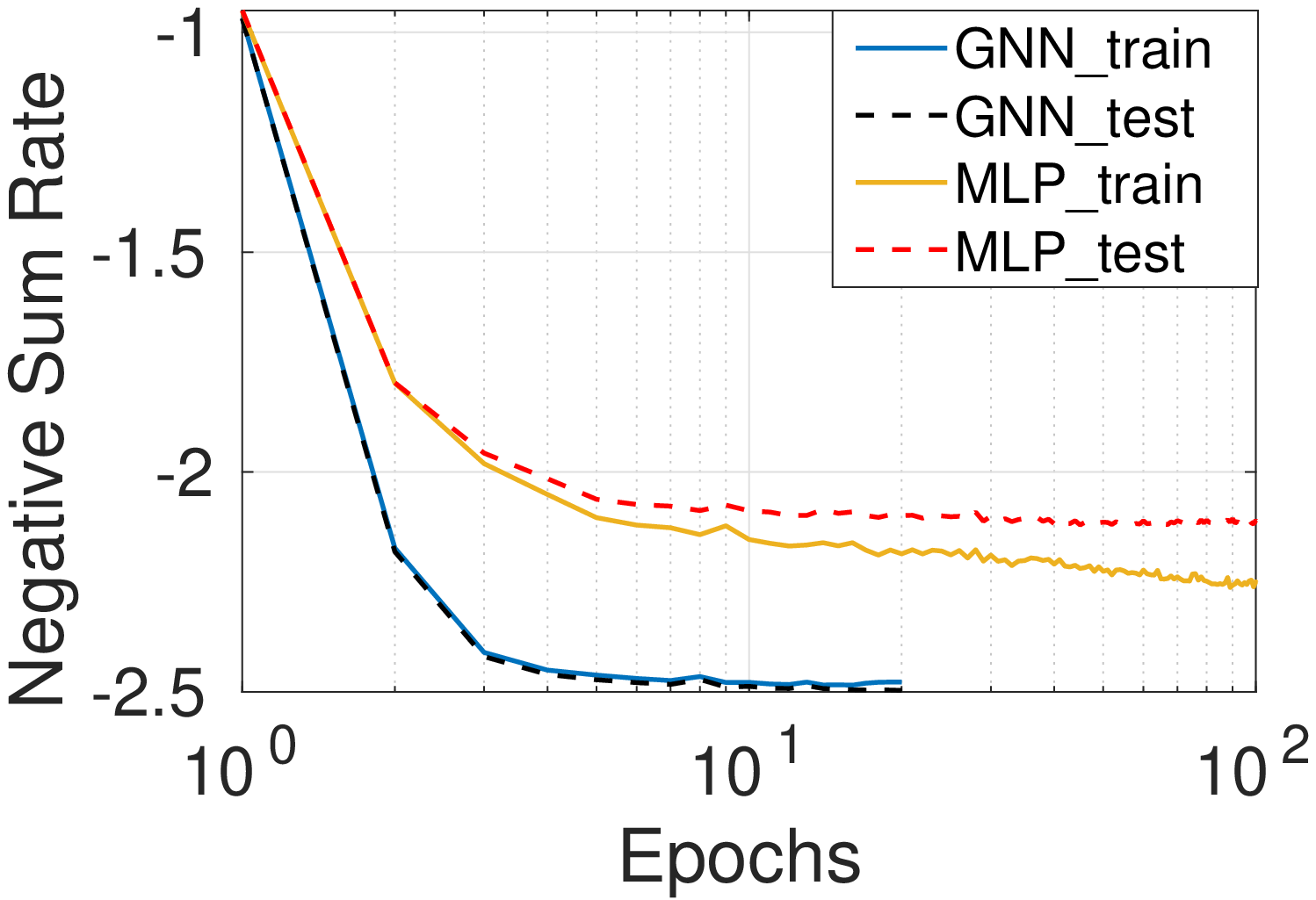}
			\label{fig:opl}
		\end{minipage}%
	}
	\centering
	\caption{Convergence and generalization of GNNs and MLPs for the $K$-user interference channel power control problem. The loss is the negative sum rate.}
	\label{fig:pc}
\end{figure*}
\begin{example} The wireless channel graph convolution network (WCGCN) proposed in \cite{shen2020graph} is a special case of MPGNNs designed for interference management in wireless networks. For the $K$-user interference channel power control problem, the update of the $k$-th node in the $j$-th layer is given as
\begin{equation}\label{eq:cgcnet}
\begin{aligned}
\bm{y}_k^{(j)} &=  \text{MAX}_{i \neq k}\left\{\text{MLP1} \left(p_k^{(j-1)},h_{i,k},h_{k,i}\right)\right\},\\
p_k^{(j)} &= \sigma\left( \text{MLP2}\left( \bm{y}_k^{(j)}, w_{k}, h_{k,k}\right)\right),
\end{aligned}
\end{equation}
where $\text{MLP1}$ and $\text{MLP2}$ are two different MLPs, and $\sigma(x) = \frac{1}{1+\exp(-x)}$ is the sigmoid function. Furthermore, $p_k^{(j)}$ is the output power at the $j$-th iteration, and $\bm{y}_k^{(j)}$ denotes the hidden state at the $j$-th iteration.
\end{example}

As most GNNs developed for communication networks are MPGNNs, for simplicity, we refer to MPGNNs as GNNs in the remainder of this paper.

\section{Main Results}
This section presents our main theoretical results. We first introduce the neural tangent kernel as a technical tool for our analysis. Then we present the main theorem, which is verified by simulations.

\subsection{Neural Tangent Kernel}
Neural tangent kernel (NTK) \cite{jacot2018neural} is a powerful tool that has been recently proposed to theoretically characterize the properties of neural networks \cite{du2019gradient,arora2019fine}. Let $\bm{u}(t) = (f(\bm{\theta}(t),\bm{x}_i))_{i \in [m]}$ be the network's output on $\bm{x}_i$ at time $t$, where $\bm{\theta}$ denotes the neural network parameters. We consider minimizing a loss function $\ell(\bm{\theta})$ by the gradient descent with an infinitesimally small learning rate.  The parameters evolve according to the following ODE
\begin{align*}
    \dot{\bm{\theta}}(t) = - \frac{\partial \ell(\bm{\theta}(t))}{\partial \bm{\theta}} = - \sum_{i=1}^m \frac{\partial \ell}{\partial f(\bm{\theta}(t), \bm{x}_i)} \frac{\partial f(\bm{\theta}(t), \bm{x}_i)}{\partial \bm{\theta}}.
\end{align*}
For the $i$-th training sample, the evolution of the neural network's output can be written as
\begin{align*}
    \dot{f}(\bm{\theta}(t), \bm{x}_i) = - \sum_{j = 1}^m  \frac{\partial \ell}{\partial \bm{u}} \left\langle \frac{\partial f(\bm{\theta}(t), \bm{x}_i)}{\partial \bm{\theta}}, \frac{\partial f(\bm{\theta}(t), \bm{x}_j)}{\partial \bm{\theta}}  \right\rangle.
\end{align*}
Thus, for the vector-form output $\bm{u}(t)$, we have 
\begin{align} \label{eq:ntk}
    \dot{\bm{u}}(t) = -\bm{H}(t) \cdot \frac{\partial \ell}{\partial \bm{u}}
\end{align}
where $[\bm{H}(t)]_{i,j} = \left\langle \frac{\partial f(\bm{\theta}(t),\bm{x}_i) }{\partial \bm{\theta} }, \frac{\partial f(\bm{\theta}(t),\bm{x}_j) }{\partial \bm{\theta} } \right\rangle$.

As the network width goes to infinity, the time-varying kernel $\bm{H}(t)$ approaches the time-invariant \emph{neural tangent kernel} $\bm{H}^* \in \mathbb{R}^{m \times m}$, where
\begin{align}\label{eq:H*}
    \bm{H}^*_{(i,j)} = \mathbb{E}_{\bm{\theta} \sim \mathcal{W}} \left\langle \frac{\partial f(\bm{\theta},\bm{x}_i) }{\partial \bm{\theta} }, \frac{\partial f(\bm{\theta},\bm{x}_j) }{\partial \bm{\theta} } \right\rangle, 
\end{align}
and $\mathcal{W}$ is a Gaussian distribution. 

It is shown in \eqref{eq:ntk} that if $\bm{H}(t)$ is a positive definite matrix, $\dot{\bm{u}} = 0$ if and only if $\frac{\partial \ell}{\partial \bm{u}} = 0$. Thus, if $\ell$ is convex, the global optimality is guaranteed at the training stage. Furthermore, as NTK bridges the neural network and kernel methods, the generalization error of the neural network can be analyzed by leveraging classic results on kernels. Different neural architectures correspond to different kernels, and thus we can analyze the convergence and generalization behaviors by studying the matrix $\bm{H}^*$ in \eqref{eq:H*}.

\subsection{Overview of Theoretical Results}
In this subsection, we study the convergence and generalization of MLPs and GNNs. We first define the generalization error:
\begin{align}\label{eq:gen_gap} \mathcal{E}_{\text{gen}} = \mathbb{E}_{\bm{x} \sim \mathcal{D}} [\ell(f(\bm{x},\bm{\theta})) - \ell(f(\bm{x},\bm{\theta}^*))],
\end{align}
where $f(\bm{x}_i, \bm{\theta})$ is the neural network's output for input $\bm{x}_i$ and parameter $\bm{\theta}$, $\ell(\cdot)$ is the loss function (performance metric), and $\mathcal{D}$ is the test distribution. We denote $\bm{\theta}$ as the obtained neural network's weights, and $\bm{\theta}^*$ as the optimal weights for the test distribution. The overall theoretical results are given in Theorem \ref{thm:overall}.

\begin{theorem}\label{thm:overall} (Informal) Assuming we train MLPs and GNNs in the NTK regime and the loss function is convex with respect to the output of the neural networks, gradient descent finds a global minimizer with a $O(1/t)$ rate, where $t$ is the number of epochs. Furthermore, suppose we learn a class of infinite-order permutation invariant functions $y = f(\bm{x}_1, \cdots, \bm{x}_n)$ with two-layer neural networks, then the convergence speed of GNNs is $O(n\log(n))$ times faster than MLPs at the training stage. At the test stage, the generalization error of MLPs is $O(n)$ times larger than GNNs.
\end{theorem}

\begin{proof}
Theorem \ref{thm:ntk_conv} in Section \ref{sec:conv} shows that gradient descent converges to the global minimizer with a $O(1/t)$ rate. The gap in convergence is shown in Theorem \ref{thm:conv} and the gap in generalization is shown in Theorem \ref{thm:gen}. 
\end{proof}

To demonstrate the power and significance of Theorem \ref{thm:overall}, we apply GNNs and MLPs to the $K$-user interference channel power control problem, where the system setting follows that of Section V.A in \cite{shen2020graph}. The results are shown in Fig. \ref{fig:pc}. With $5$ users, both MLPs and GNNs achieve a good training loss and a similar test error. However, with $20$ users, MLPs have difficulty in training while GNNs converge quickly. Furthermore, there is a large performance gap between MLPs and GNNs with $20$ users at the test stage.

\section{Detailed Analysis}
This section presents detailed derivations of Theorem \ref{thm:overall}.
\subsection{Convergence}\label{sec:conv}
\begin{figure*}[htbp]
	\centering
	\subfigure[Learning a permutation invariant function on a $1$-node graph.]{
		\begin{minipage}[t]{0.22\linewidth}
			\centering
			\includegraphics[width=1\linewidth]{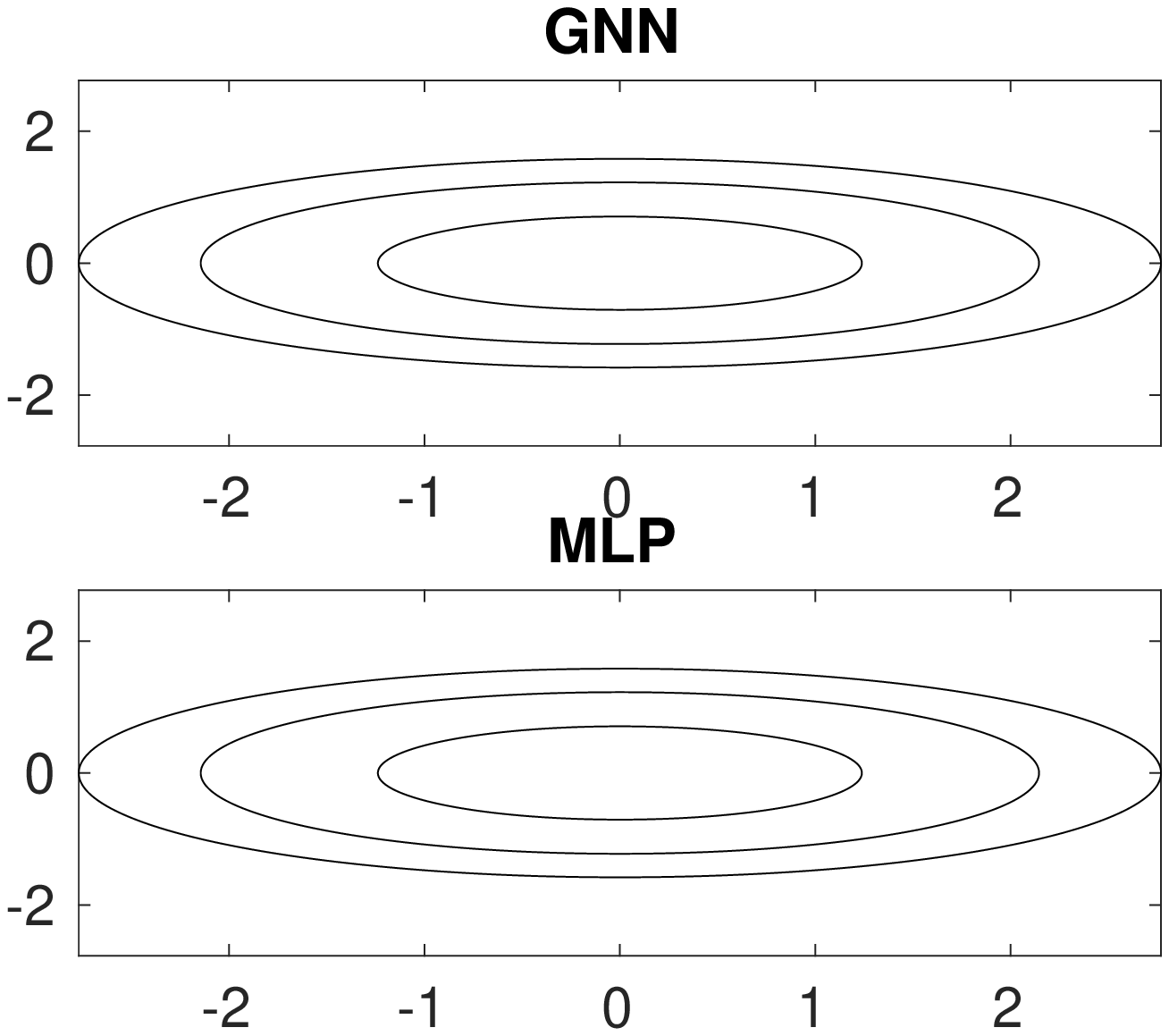}
		\end{minipage}%
	}%
	\hspace{.1in}
	\subfigure[Learning a permutation invariant function on a $5$-node graph.]{
		\begin{minipage}[t]{0.22\linewidth}
			\centering
			\includegraphics[width=1\linewidth]{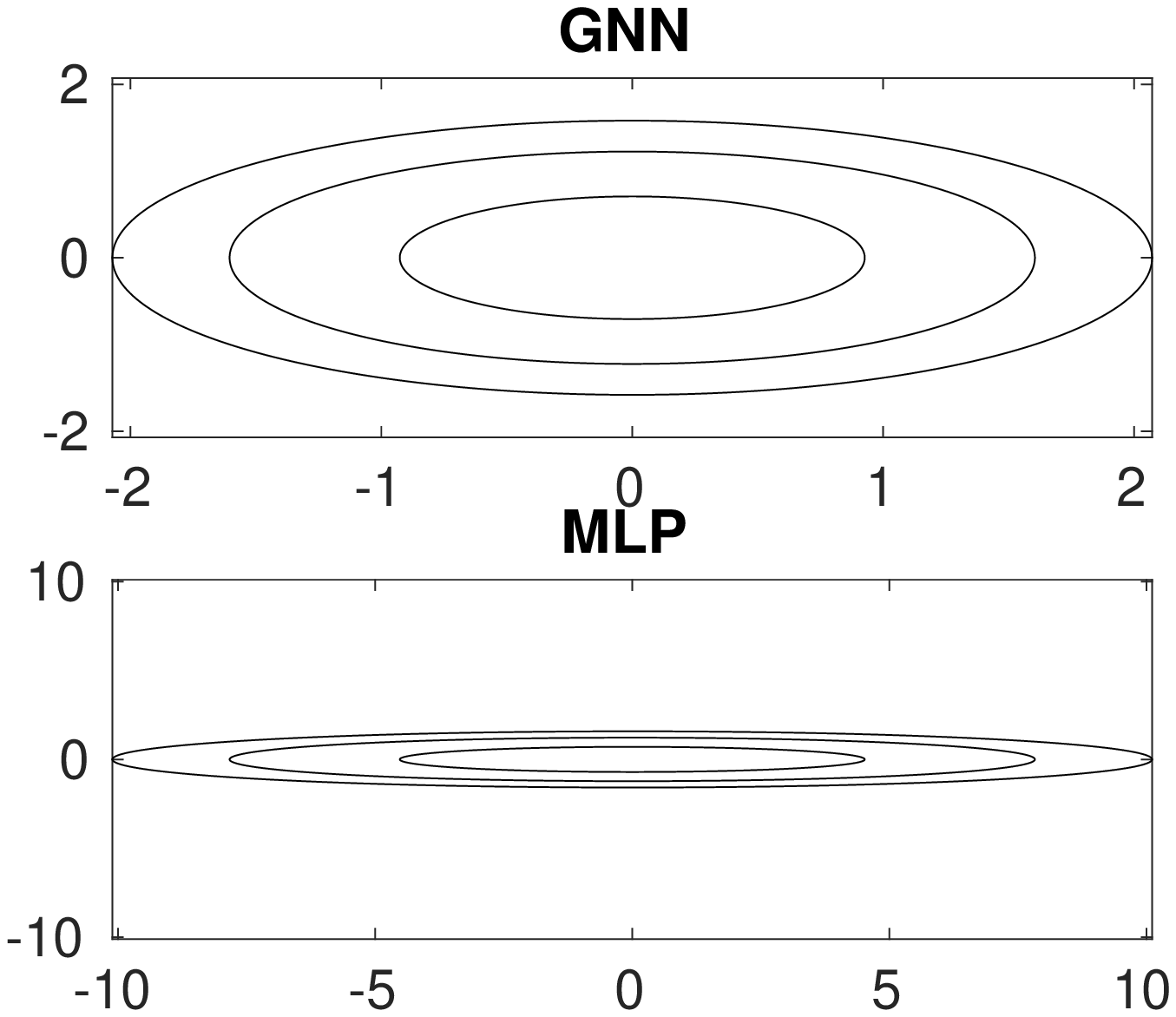}
		\end{minipage}%
	}%
	\hspace{.1in}
	\subfigure[Learning a permutation invariant function on a $20$-node graph.]{
		\begin{minipage}[t]{0.22\linewidth}
			\centering
			\includegraphics[width=1\linewidth]{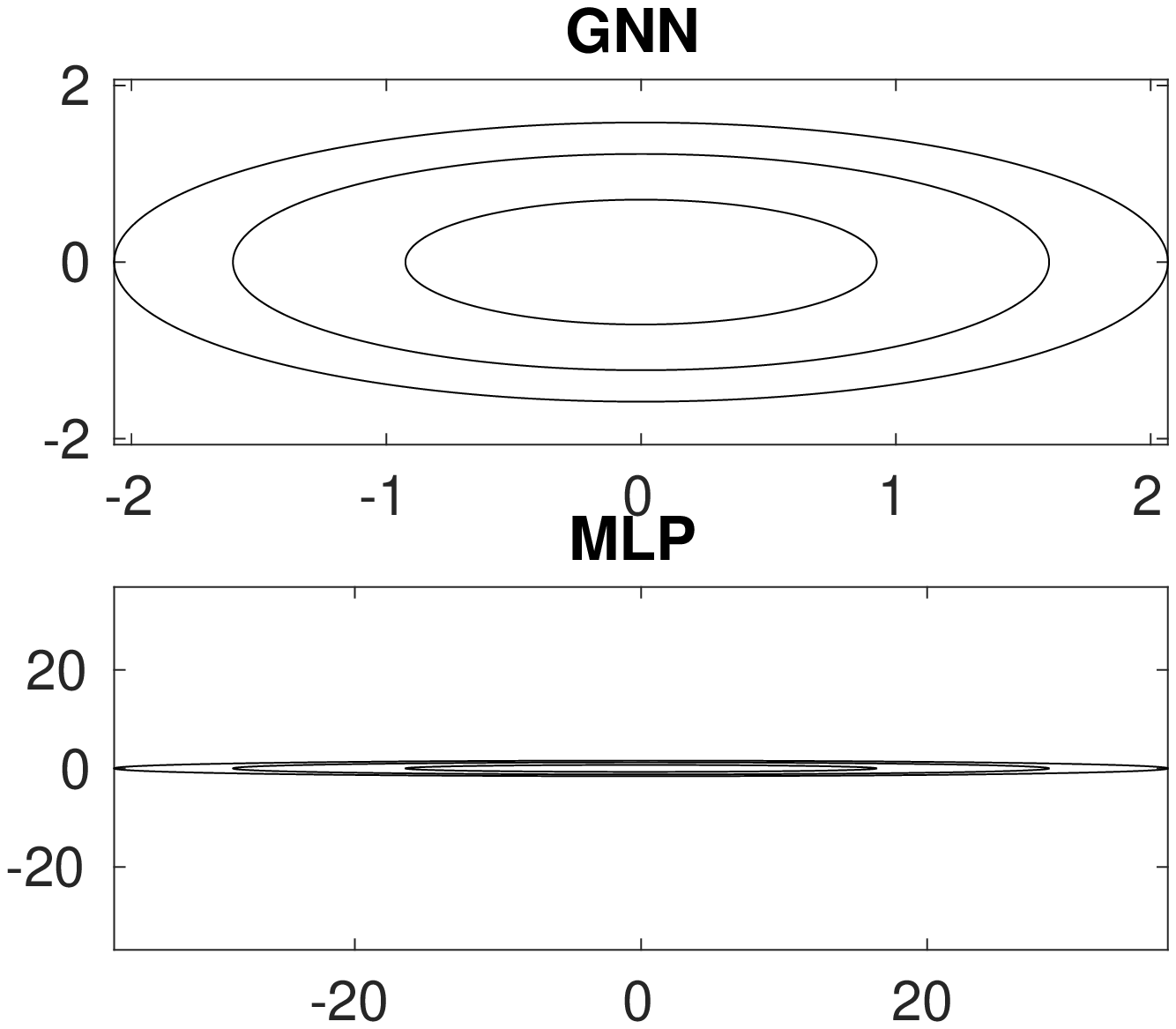}
		\end{minipage}%
	}%
	\hspace{.1in}
	\subfigure[Learning an non-permutation invariant function on a $2$-node graph.]{
		\begin{minipage}[t]{0.22\linewidth}
			\centering
			\includegraphics[width=1\linewidth]{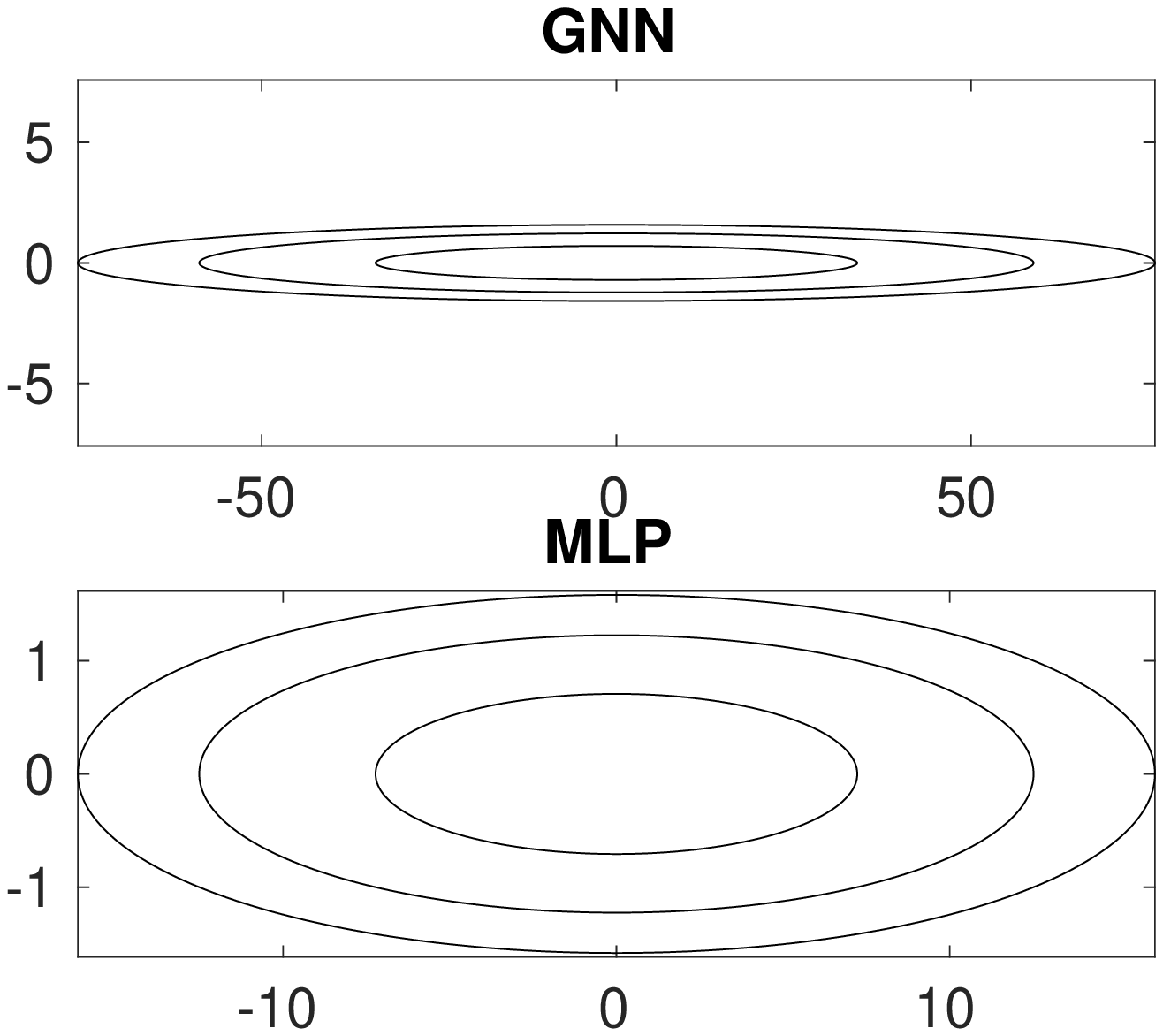}
		\end{minipage}%
	}%
	\caption{An illustration of the optimization landscape for two-layer infinitely wide MLPs \cite{arora2019fine} and GNNs \cite{du2019graph}. The length of the long axis represents the conditional number. A more isotropic plot means better convergence. The plots are generated from $300$ i.i.d. Gaussian samples. In (a)(b)(c), the target function is a linear permutation invariant function while in (d) the target function is non-permutation invariant.}
	\label{fig:landscape}
\end{figure*}
In this subsection, we study the convergence of GNNs and MLPs. We first discuss the convergence rate of neural networks under convex loss functions.

\begin{theorem}\label{thm:ntk_conv} (Global convergence of neural networks with convex loss functions)
Assume $\lambda_{\min}(\bm{H}^*) > 0$, for a convex and differentiable loss function, the convergence rate is given by $\ell(\bm{u}(t)) - \ell(\bm{u}^*) \leq \frac{\|\bm{u}(0) - \bm{u}^*\|_2^2}{2\kappa t}$,
where $\kappa = \max_{t} \frac{\|\bm{H}^* \bm{x}(t)\|_2 \|\bm{x}(t)\|_2 }{\bm{x}(t) \bm{H}^* \bm{x}(t)} \leq \frac{\lambda_{\max}(\bm{H}^* )}{\lambda_{\min}(\bm{H}^* )}$, and $\bm{x}(t) = \frac{\partial \ell}{\partial \bm{u}}$. For the squared loss function, we have $\bm{y} - \bm{u}(t) = \exp(- \bm{H}^* t) \bm{y}$.
\end{theorem}

In the context of communication networks, the universal approximation theorem of neural networks has been adopted to justify the application of deep learning, e.g., for resource management \cite{sun2018learning,eisen2019optimal}, and MIMO detection \cite{qiang2019deep}, but nothing has been said about whether the optimal weights of the neural network can be found via training. Theorem \ref{thm:ntk_conv} makes a further step by showing that the gradient descent algorithm is able to find such a neural network, given that the loss function is convex. For non-convex loss functions, there exist bad stationary points as shown in \cite{song2021supervise}.

In Theorem \ref{thm:ntk_conv}, the convergence heavily depends on the alignment between the eigenvectors of the NTK matrix ($\bm{H}^*$) and the target vector ($\bm{y}$ or $\bm{x}(t)$). As $\bm{H}^*$ is a structured random matrix, it is possible to have a fine-grained analysis on the convergence by studying its eigenvectors, which reveals the gap between GNNs and MLPs.

\begin{theorem} \label{thm:conv} (GNNs converge faster than MLPs)
    Suppose we train two-layer neural networks in the NTK regime with squared loss function, and the target function is $\bm{u}^*(\bm{x}_1, \cdots \bm{x}_n) = \sum_{i = 1}^n (\bm{\beta}^T \bm{x}_i)^{p}$. Then for GNNs, we have $\|\bm{u}(t) - \bm{u}^*\|_2 \leq \exp\left(- c_{p, \sigma}\left(\sum_{i=1}^n \lambda_i\right)t\right)  \|\bm{\beta}\|_2^{p}$. For MLPs, we have $\|\bm{u}(t) - \bm{u}^*\|_2 \leq n \exp\left(- c_{p, \sigma} \lambda_{\min}t\right)   \|\bm{\beta}\|_2^{p}$, where $c_{p,\sigma}$ is a constant related to $p$ and the activation function $\sigma$. In addition, $\lambda_1, \cdots, \lambda_n$ are non-negative constants regarding the input data, and $\lambda_{\min} = \min_{i}(\lambda_1, \cdots, \lambda_n)$.
\end{theorem}

One concluding message from Theorem \ref{thm:conv} is that a proper neural architecture improves the convergence rate. Specifically, both the permutation invariance property and the activation function influence the convergence upper bound. Thus, the convergence rate can be improved with a specialized design of these components.

\textbf{Impact of permutation invariance}: In the bound, we see that the convergence rate of MLPs depends on $\lambda_{\min}$, while the rate of GNNs depends on $\sum_{i=1}^n \lambda_i$ if the target function $\bm{u}^*$ is permutation invariant. This is because the architecture of GNNs improves the optimization landscape. We plot the conditional number of infinitely wide MLPs \cite{arora2019fine} and GNNs \cite{du2019graph} in Fig. \ref{fig:landscape}. We see that as $n$ grows, the conditional number of MLPs becomes larger while that of GNNs remains the same. This shows that as the problem size grows, the convergence of MLPs slows down dramatically while that of GNNs does not. This impedes MLPs to achieve a low training loss at the training stage and is the main reason why GNNs can achieve near-optimal performance with a large number of users while MLPs fail to do so.

\textbf{Impact of activation functions and unrolling}: The activation function influences the convergence by controlling the term $c_{p,\sigma}$. The constant $c_{p, \sigma}$ is to measure the similarity between the activation function $\sigma$ and the target function $\bm{u}^*$. For example, denoting $\sigma_1(x) = x^2$ and $\sigma_2(x) = \max(0,x)$, we have $c_{2,\sigma_1} = 1$ and $c_{2,\sigma_2} = \frac{1}{2\pi}$. This implies that if the target function is quadratic, the neural network with quadratic activation will converge faster than neural networks with other activation functions. In the deep unrolling methods \cite{he2019model}, we can view the operations borrowed from classic algorithms as activation functions. Thus, for tasks where precise mathematical modelling is available, deep unrolling can accelerate the training of both MLPs and GNNs \cite{chowdhury2021unfolding}.

\textbf{No free lunch}: The previous discussion shows that GNNs converge faster than MLPs for permutation invariant target functions. Nevertheless, when the target function is not permutation invariant, GNNs may have a worse conditional number than MLPs as shown in Fig. \ref{fig:landscape} (d). This implies that GNNs are good at learning permutation invariant functions while performing poorly in learning non-permutation invariant functions.

In the context of communication networks, as discussed in Section \ref{sec:pre}, permutation invariance commonly exists, so GNNs stand out as a promising neural architecture.

\subsection{Generalization}
In this section, we analyze the generalization of GNNs and MLPs, based on \cite{arora2019fine,du2019graph,xu2019what}. We begin with a classic result on the generalization error of kernel methods.
\begin{theorem}\label{thm:gen_all} \cite{bartlett2002rademacher,du2019graph} Given $m$ training data $\{\bm{x}_i, y_i\}_{i=1}^m$ drawn i.i.d. from the underlying distribution $\mathcal{D}$. Consider a loss function $\ell:\mathbb{R} \times \mathbb{R} \rightarrow [0,1]$ that is 1-Lipschitz in the first argument. With probability $1-\delta$, the population loss of infinitely wide neural networks is bounded by $\mathcal{E}_{\text{gen}} = \mathbb{E}_{\bm{x} \sim \mathcal{D}} [\ell(f(\bm{x},y)] = \mathcal{O}\left( \frac{\sqrt{\bm{y}^T(\bm{H}^*)^{-1} \bm{y} \cdot \trace(\bm{H}^*)}}{m}+\sqrt{\frac{\log(1/\delta)}{m}}\right)$.
\end{theorem}

\textbf{Higher sample efficiency implies better generalization}: Recently, there is a growing interest in sample-efficient neural architectures for solving communication problems \cite{shen2018lora,lee2019graph,sun2020reducing}. Theorem \ref{thm:gen_all} suggests that the generalization error $\mathcal{E}_{\text{gen}}$ is inversely propositional to the number of training data points $m$. Thus, a neural architecture with a higher sample efficiency results in a smaller test error.

The next theorem analyzes the generalization error of MLPs and GNNs, which is based on \cite{arora2019fine,du2019graph}.
\begin{theorem}\label{thm:gen} (GNNs generalize better than MLPs)
    Suppose we train two-layer neural networks in the NTK regime with squared loss function, and the target function is $\bm{u}^*(\bm{x}_1, \cdots \bm{x}_n) = \sum_{i = 1}^n (\bm{\beta}^T \bm{x}_i)^{p}$. Then for GNNs, with probability $1-\delta$, we have $\mathcal{E}_{\text{gen}}^{\text{GNN}} \leq \mathcal{O}\left(\frac{c_{\sigma,p} \|\bm{\beta}\|_2^{p} }{m} + \sqrt{\frac{\log(1/\delta)}{m}}\right)$. For MLPs, with probability $1-\delta$,$ \mathcal{E}_{\text{gen}}^{\text{MLP}} \leq \mathcal{O}\left(\frac{n c_{\sigma,p} \|\bm{\beta}\|_2^{p} }{m} + \sqrt{\frac{\log(1/\delta)}{m}}\right)$, where $c_{p,\sigma}$ is a constant related to the activation function $\sigma$ and the degree $p$ and activation function $\sigma(\cdot)$.
\end{theorem}

Similar to the convergence results, the generalization error is also influenced by the permutation invariance and activation function, and GNNs are superior if the target function is permutation invariant.

\begin{figure*}[htbp]
	\centering
	\subfigure[MLPs.]{
		\begin{minipage}[t]{0.4\linewidth}
			\centering
			\includegraphics[width=1\linewidth]{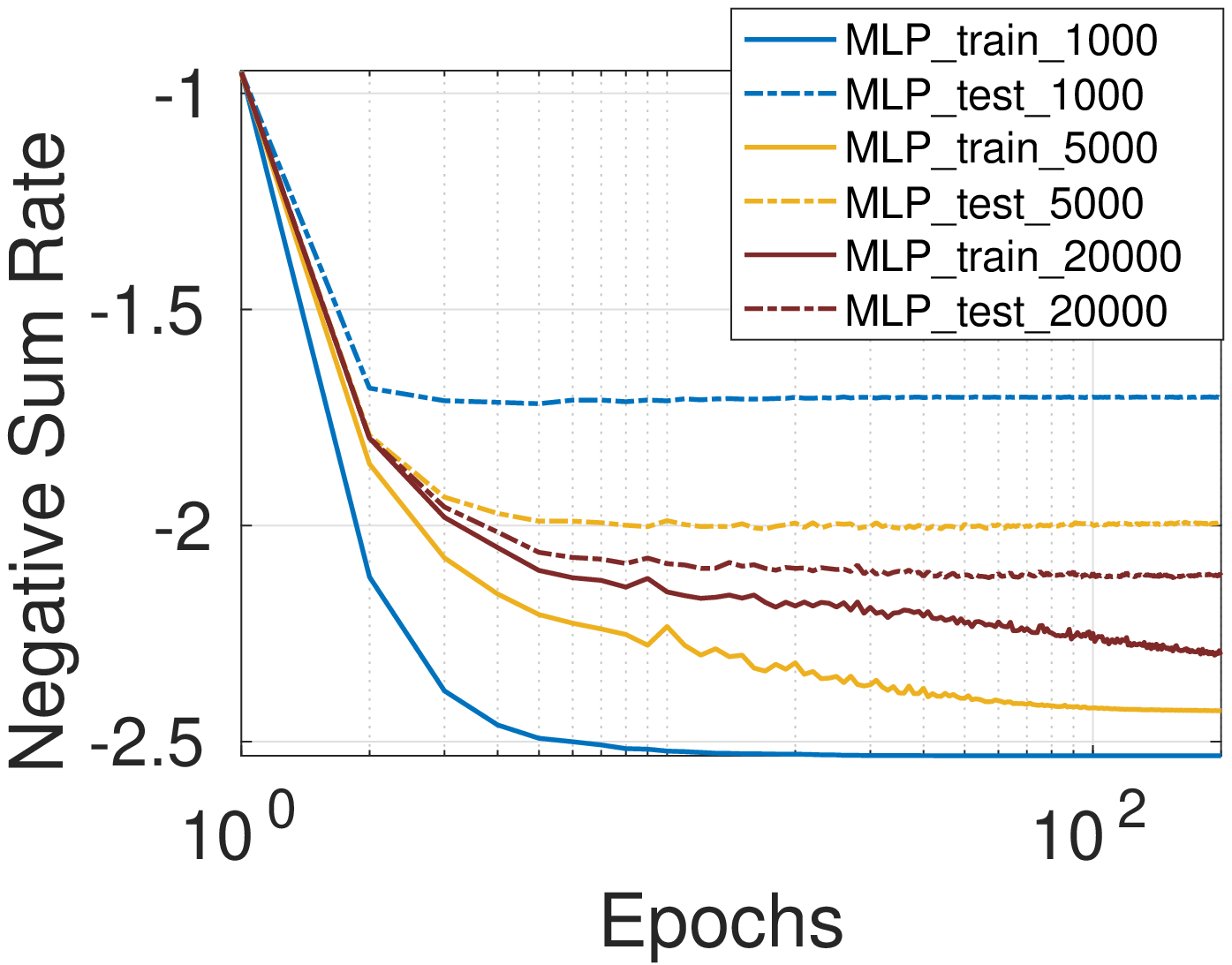}
			\label{fig:e2e}
		\end{minipage}%
	}%
	\subfigure[GNNs.]{
		\begin{minipage}[t]{0.4\linewidth}
			\centering
			\includegraphics[width=1\linewidth]{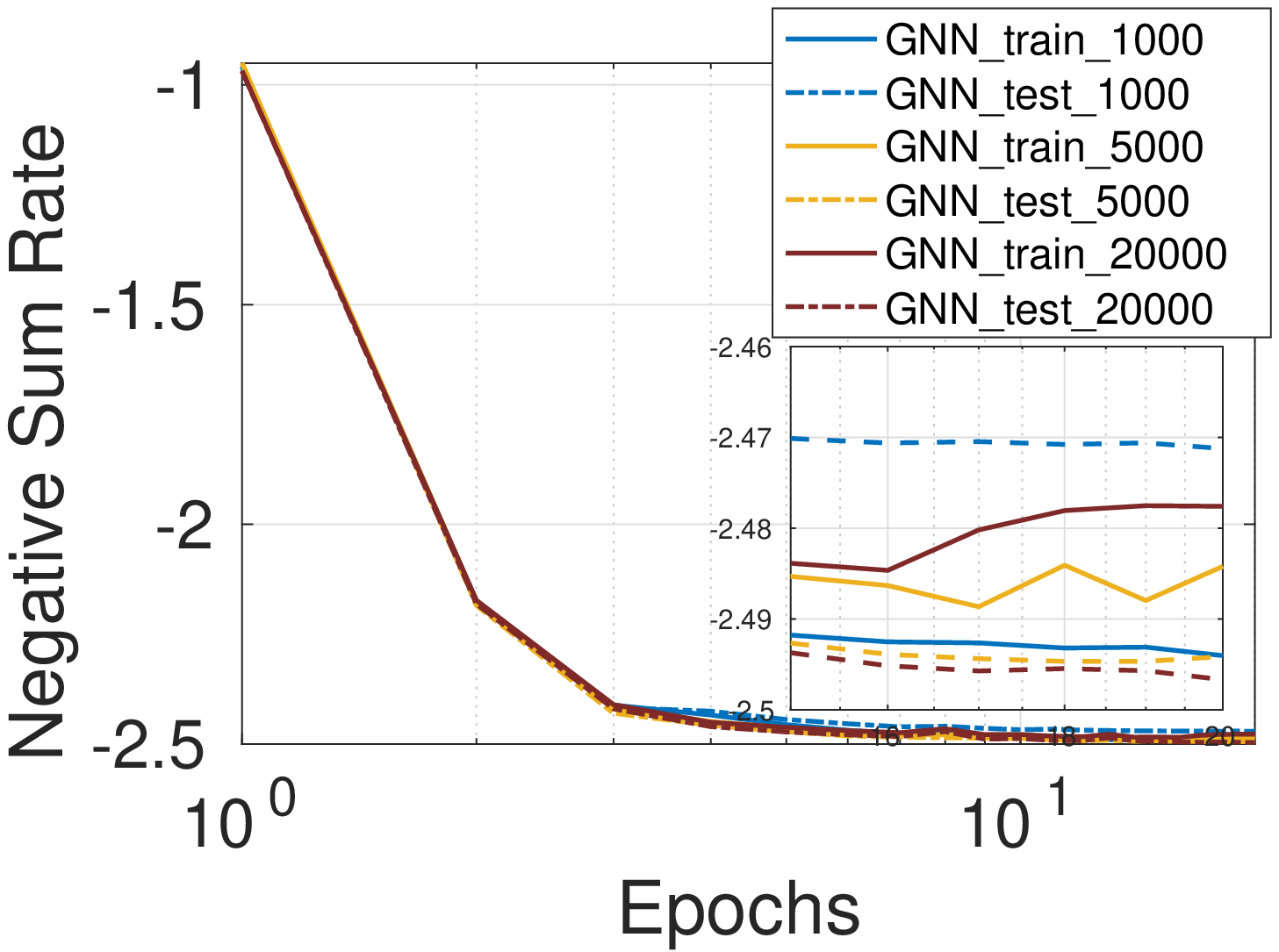}
			\label{fig:opl}
		\end{minipage}%
	}%
	\centering
	\caption{Convergence and generalization of MLPs and GNNs with different numbers of training samples. The loss value is the negative sum rate. Different curve represents different numbers of training samples. For example, MLP\_train\_20000 is the training loss curve of the MLP with $20000$ training samples.}
	\label{fig:tradeoff}
\end{figure*}

\subsection{Can MLPs Match GNNs with Sufficient Data?}
Theorem \ref{thm:gen} indicates that with a finite number of training samples (i.e., $m$), there is a substantial gap in the generalization performance between MLPs and GNNs, proportional to the number of nodes (i.e., $n$). In some problems, training data can be easily generated from simulations or ray tracing \cite{alkhateeb2019deepmimo}, which allows us to have sufficiently many training samples. A natural question to ask is: \emph{Given a sufficiently large amount of data, can MLPs perform as well as GNNs?} The answer is yes if we can train the model for infinitely long time. In practice, however, a larger dataset will make the training more difficult. This is because the smallest eigenvalue of $\bm{H}^* \in \mathbb{R}^{m \times m}$ shrinks as the number of training samples $m$ increases. For MLPs, the smallest eigenvalue of $\bm{H}^*$ will be reduced when the number of training samples increases, which increases the training difficulty. For GNNs with permutation invariant target functions, this phenomenon is not obvious due to the improvement of the landscape by neural architectures. In Fig. \ref{fig:tradeoff}, we follow the system setting of Section V.A in \cite{shen2020graph} and test MLPs and GNNs in the $K$-user interference channel power control problem with $K=20$. It shows that more training data slow down the convergence of MLPs to a large extent, which deteriorates the performance. Thus, GNNs are superior to MLPs even if a very large number of samples are available for training.

\section{Conclusions}
This paper theoretically investigated the importance of neural architectures when applying deep learning in communication networks. We proved that by exploiting the permutation invariance property, GNNs converge faster and generalize better than MLPs. For future directions, it is interesting to extend the analysis to other neural architectures, which will lead to a systematic and principled design of neural architectures in the area of machine learning for communication.

\bibliographystyle{ieeetr}
\bibliography{ref}

\end{document}